\documentclass[11pt]{article}
\usepackage[margin=1in]{geometry}
\usepackage{amsthm,amsfonts,amsmath,amssymb}
\usepackage{mathtools}
\usepackage{enumerate,graphicx}
\usepackage{xspace}
\usepackage{boxedminipage}
\usepackage{url}
\usepackage{soul}
\usepackage{accents}
\usepackage{multirow}
\usepackage{algorithmicx}
\usepackage{algpseudocode}
\usepackage{booktabs}
\usepackage{enumitem}
\usepackage{fullpage}
\usepackage{tikz}
\usepackage{natbib}
\usepackage{tikz-cd}
\usepackage{thmtools}
\usepackage{thm-restate}
\usepackage{array}
\usepackage{booktabs}
\usepackage{threeparttable}
\usepackage{tcolorbox}
\usepackage{ulem}
\usepackage{subcaption}
\usepackage{hyperref}
\hypersetup{
    colorlinks=true,
    linkcolor=violet,
    filecolor=magenta,      
    urlcolor=cyan,
    citecolor=blue,
    pdffitwindow=true,
}\usepackage[capitalize, nameinlink]{cleveref}

\usepackage{algorithm}
\usepackage{algpseudocode}

\usepackage{color-edits}[showdeletions]

\newtheorem*{theorem*}{Theorem}

\newtheorem{lemma}{Lemma}

\newtheorem{claim}[lemma]{Claim}

\theoremstyle{definition}
\newtheorem{remark}[lemma]{Remark}

\newtheorem{observation}[lemma]{Observation}
\newtheorem{definition}[lemma]{Definition}

\newcommand{\eps}{\varepsilon}

 \usepackage{booktabs}
 \usepackage{tabularx}
 \usepackage{array}

\newcolumntype{Y}{>{\raggedright\arraybackslash}X}

\title{A Linear-Time Approximation Scheme for the Densest Subgraph Problem}

\author{  Elena Grigorescu\thanks{University of Waterloo. Email: \textsf{elena.grigorescu@uwaterloo.ca}} \quad  Mehrshad Taziki\thanks{ETH Zürich. Email: \textsf{mtaziki@ethz.ch}} }

\begin{document}

\maketitle
\begin{abstract}
    In the undirected \emph{Densest Subgraph Problem (DSG)} the goal is to output a subset  $S$ of vertices of a given graph $G$ that maximizes the quantity $|E(S)|/|S|$, where $E(S)$ is the set of edges in the subgraph induced by $S$. The problem is  well studied in both theory and practice, and it admits natural efficient exact algorithms, as well as near-linear time algorithms with a $(1-\eps)$ approximation ratio. However, all previously-known approximation schemes incur logarithmic factors in the size of the graph or other parameters of the graph. This raises the question of whether a linear time $(1-\varepsilon)$-approximation can be obtained for all $\varepsilon>0$.  We answer this question affirmatively by providing a $(1-\eps)$-approximation algorithm running in time 
    \begin{align*}
        O\left(\frac{n+m}{\varepsilon^3}\log \frac{1}{\varepsilon}\right)
    \end{align*}
    where $m$ and $n$ are respectively the number of edges and vertices of $G$. To the best of our knowledge, this is the first truly linear-time approximation scheme for the problem (when $\eps>0$ is a constant). Our algorithm uses assignments arising from a flow-based formulation together with a structural carving lemma. This lemma allows us  to progressively carve ``sparse'' parts of the graph while nearly preserving the densest subgraph, allowing us to shift heavy computations to smaller instances,  which eventually yields the mentioned runtime.

    Our framework also yields a  $(1/2 -\varepsilon)$-approximation for the \emph{Densest At-Least-$k$ Subgraph Problem}, where in addition to maximizing the density, we require the subgraph to have at least $k$ vertices. Our algorithm runs in time
    \begin{align*}
        O\left( \frac{(n+m) \log^2 n \log \frac{1}{\varepsilon}}{\varepsilon}  \right).
    \end{align*}
    This nearly matches the known $1/2$ approximation hardness while  running in near-linear time.
\end{abstract}

\section{Introduction}

The undirected \emph{Densest Subgraph Problem (DSG)} is fundamental to network design, combinatorial optimization, and graph mining. Given a graph $G=(V,E)$, the goal of DSG is to output a subset $S$ of vertices of $G$ that maximizes the quantity $d(S)=|E(S)|/|S|$, called \emph{the density of $S$}, where $E(S)$ is the set of edges in the subgraph induced by $S$. It is long known that the problem is solvable exactly and efficiently in time $\tilde{O}(nm)$ \cite{goldberg1984, gallo1989} via reductions to max flow, where $m=|E(G)|$ and $n=|V(G)|$. Nevertheless, the cost of exact flow computations can be prohibitive on the massive graphs that motivate many applications of dense subgraph discovery. Later Asahiro et al. \cite{asahiro1996} and Charikar \cite{charikar2000} gave a greedy peeling algorithm that repeatedly removes the vertex with the lowest degree from the graph and returns the densest solution seen. This simple procedure obtains a $\frac{1}{2}$-approximation solution for DSG in time $O(n+m)$.

Substantial recent works have thus aimed to give scalable algorithms that have near optimal approximation guarantees. These algorithms can be viewed as a mix of  \emph{peeling} algorithms, maximum flow and linear programming (LP) formulations, and continuous optimization approaches, including multiplicative weights updates (MWU) algorithms. Boob et al. \cite{boob2020} adapted the greedy peeling procedure into an efficient iterative variant. Later Chekuri, Quanrud, and Torres \cite{chekuri2022} proved a $(1-\varepsilon)$-approximation ratio for this algorithm in time $O\left( \frac{(n+m)\log n}{\varepsilon^2} \cdot \frac{\Delta}{\lambda^*} \right)$, where $\Delta$ is the maximum degree of the graph and $\lambda^*$ is the density of the densest subgraph.

The influential result of Charikar \cite{charikar2000} also brought up the use of  linear and convex programming  tools to address the densest subgraph problem, which were further exploited by Bahmani, Goel, and Munagala~\cite{bahmani2014} who developed a primal--dual algorithm achieving a $(1-\varepsilon)$-approximation in $O\left(\frac{(n+m) \log n}{\varepsilon^2} \right)$. More recently, Chekuri, Quanrud, and Torres~\cite{chekuri2022} gave a flow-based $(1-\eps)$-approximation algorithm motivated by Charikar's LP formulation running in time $O\left(\frac{(n+m) \log^2 n}{\varepsilon}\right)$. Nguyen and Ene~\cite{nguyenEne2024} obtained additional $(1-\varepsilon)$-approximation algorithms based on multiplicative weights and area convexity, running respectively in times $O\left( \frac{m \log \Delta \log m}{\varepsilon^2} \right)$ and  $O\left( \frac{m \log \Delta \log m \log (1/\varepsilon)}{\varepsilon} \right)$.

A further set of convex optimization approaches includes extensions of Charikar's LP to quadratic programming formulations, initially solved using  the Frank-Wolfe algorithm in time $O(m\Delta/\eps^2)$ \cite{danisch2017}, and later improved by Harb, Quanrud and  Chekuri \cite{harb2022, harbESA2023} to $O(\sqrt{m\Delta}/\eps)$, who also show that the flow-based approach in \cite{chekuri2022} is in fact equivalent to the Frank-Wolfe algorithm.

In applications where the output is intended to represent a substantial group or community, it is natural to impose a constraint on its size. Requiring the output to contain exactly $k$ vertices gives the well-known \emph{Densest $k$-Subgraph Problem}, which is substantially harder to approximate than unconstrained \textsc{DSG} (see, e.g., \cite{KortsarzPeleg93,asahiro1996,FeigeKortsarzPeleg01,AsahiroHassinIwama02,Khot06,GoldsteinLangberg09,BhaskaraEtAl10,BravermanEtAl17,Manurangsi17,JonesEtAl23SoS}). Two natural size-bounded relaxations are the \emph{Densest At-Least-$k$ Subgraph Problem}, in which the output is required to contain at least $k$ vertices (see, e.g., \cite{Andersen2009,khuller2009,GoldsteinLangberg09,ChenEtAl15Size,ChenEtAl16Size,Manurangsi18,chekuri2022,LaekhanukitManurangsiTrabelsi26}), and the \emph{Densest At-Most-$k$ Subgraph Problem}, in which the output is required to contain at most $k$ vertices (see, e.g., \cite{Andersen2009,khuller2009,GoldsteinLangberg09,ChenEtAl16Size}). Unlike \textsc{DSG}, all of these size-constrained variants are NP-hard~\cite{AsahiroHassinIwama02,Andersen2009,khuller2009}.

Several other variants of the DSG problem and notions of density have also been considered in the literature. The directed densest subgraph problem was proposed by Kannan and Vinay \cite{kannan1999analyzing}, where the goal is to find subsets $S, T$ of the vertices maximizing $\frac{|E[S,T]|}{\sqrt{|S||T|}}$.  Charikar \cite{charikar2000} obtained an exact algorithm for directed DSG based on an LP relaxation, as
well as a faster 2-approximation. Later, Khuller and Saha \cite{khuller2009} obtained a faster flow-based exact algorithm. Andersen \cite{andersen2010local} presented a local algorithm to find dense directed subgraphs. 

We refer the interested reader to the following surveys \cite{Farago2019, lanciano2024survey, Lee2010}.

\subsection{Our Results}

Prior results on DSG achieve running times that are nearly linear in the graph size. However, they retain logarithmic dependence on $n$ or other graph parameters \cite{bahmani2014,chekuri2022,nguyenEne2024}. This motivates the following question: 

\begin{quote}
\textit{Is there an algorithm that achieves a $(1-\varepsilon)$-approximation while running in time that is linear in the size of the graph for all fixed $\varepsilon > 0$?}
\end{quote}

We answer this question affirmatively by proving the following theorem.

\begin{restatable}{theorem}{thmMainUndir}\label{thm:undir}
    There exists a $(1-\varepsilon)$-approximation algorithm for the undirected DSG problem, running in time $O(\frac{n+m}{\varepsilon^3} \log \frac{1}{\varepsilon} )$. 
\end{restatable}

In particular, for every fixed $\varepsilon>0$, the running time in \cref{thm:undir} is linear in the input size.  To our knowledge, this is the first approximation scheme for DSG with this property.

Our second result concerns the Densest At-Least-$k$ Subgraph Problem (\textsc{DalkSP}).  Here the feasible solutions are required to contain
at least $k$ vertices, and the optimal value is
\[
    \lambda^*_{\geq k}(G)
    :=
    \max_{\substack{S\subseteq V\\ |S|\geq k}}
    \frac{|E[S]|}{|S|}.
\]
Unlike unconstrained DSG, this problem is NP-hard.  Nevertheless, our framework yields an approximation
arbitrarily close to $1/2$ in near-linear time.

\begin{restatable}{theorem}{thmatleast}
    \label{thm:atleast}
    There exists a $(\frac{1}{2}-\eps)$-approximation algorithm for the DalkSP problem running in time $O\left( \frac{(n+m)\log^2 n \log \frac{1}{\varepsilon}}{\varepsilon}  \right)$.
\end{restatable}

Moreover, assuming the Small-Set Expansion Hypothesis, \textsc{DalkSP} is NP-hard to approximate within a factor of $2-\varepsilon$ for any constant  $\varepsilon>0$, making the
$1/2$ approximation threshold essentially
tight~\cite{Manurangsi18}. To the best of our knowledge, this is the first near-linear algorithm for this problem that reaches the approximation upper bound $1/2$. 

\subsection{Preliminaries}
We begin by fixing some notation that will be used throughout the paper.

\paragraph{Notation}  For an undirected graph $G = (V,E)$, by $\delta_G(S)$ we denote the edges of $G$ with exactly one endpoint in $S$, and by $G[S]$ we denote the induced subgraph on the set of vertices $S$. For the sake of brevity when $G$ is clear from the context we use $\delta(S)$ to denote $\delta_G(S)$ and for $S \subseteq T$ define $\delta_T(S)$ to denote $\delta_{G[T]}(S)$. Similarly, $E[S]$ is used to denote the edges with both end points in $S$. We use $\nu(S)$ to denote $\delta(S) \cup E[S]$, i.e. the edges with  at least one endpoint inside $S$. For a vertex $v \in G$, $\mathrm{deg}_G(v) = |\delta_G(v)|$ denotes the degree of $v$ in $G$ and $\mathrm{deg}_S(v)$ denotes $\mathrm{deg}_{G[S]}(v)$. 

When $G$ is a directed graph, $\delta^+(S), \delta^-(S)$ respectively denote the outgoing and incoming edges of the set $S$. For a vertex $v \in G$, $\mathrm{deg}_G^+(v) = |\delta^+_G(v)|$ denotes the outgoing degree and $\mathrm{deg}_G^-(v) = |\delta_G^-(v)|$ denotes the incoming degree of $v$.

Given a graph $G$, $E[G]$ denotes the edges of $G$ and $V[G]$ denotes the vertices of $G$. We use $S^*(G)$ to denote an arbitrary densest subgraph in $G$ and $\lambda^*(G)$ denotes the density of $S^*(G)$. Similarly, $S^*_{\geq k}(G)$ denotes an arbitrary densest at-least-$k$ subgraph in $G$ and $\lambda_{\geq k}^*(G)$ denotes the density of $S^*_{\geq k}(G)$.

We use $\log$ to denote the base $2$ logarithm.

\paragraph{Core Decomposition of a Graph} A critical ingredient of our procedure is the concept of ``$\kappa$-core'' of a graph introduced by Seidman~\cite{seidman1983network}. For a threshold $\kappa\geq 0$, the $\kappa$-core of a graph is the
largest induced subgraph in which every vertex has degree at least
$\kappa$.

\begin{definition}
    The $\kappa$-core of a graph $G=(V,E)$ is the maximal subset $S \subseteq V$ such that the degree of every vertex in $G[S]$ is at least $\kappa$.
\end{definition}

The $\kappa$-core is unique and can be obtained by repeatedly deleting vertices whose current degree is smaller than $\kappa$.  Moreover, the cores for all thresholds can be computed simultaneously in $O(n+m)$ time using a smallest-last peeling ordering
\cite{matula1983smallest,batagelj2003cores}.

\begin{lemma}
    \label{lem:core}
    Given a graph $G = (V,E)$ and a real number $\kappa$, there is an algorithm that finds the $\kappa$-core of $G$ in $O(n+m)$ time. We denote this algorithm by $\textsc{Core}(G, \kappa)$.
\end{lemma}

We use the following observations during the analysis of our algorithm.

\begin{observation}
    Let $H$ be the $\kappa$-core of a graph $G$, then,
    \begin{align*}
        |E[H]| \geq \frac{\kappa}{2} |V(H)|.
    \end{align*}

    Moreover, if $\kappa \leq \lambda^*(G)$, then $S^*(G) \subseteq V(H)$.
\end{observation}
\begin{proof}
First, using the handshaking lemma we observe,
\begin{align*}
    2 |E[H]| = \sum_{v \in V(H)} \deg_H(v) \geq  \sum_{v \in V(H)} \kappa = \kappa|V(H)| 
\end{align*}
where the second inequality simply uses the definition of a $\kappa$-core. This proves the first claim.

To see the second claim, note that any $v \in S^*(G)$ satisfies $|\delta_{S^*(G)}(v)| \geq \lambda^*(G)$. Otherwise, the graph $S^*(G) \setminus v$ has $|S^*(G)| -1$ vertices and strictly more than $|E[S^*(G)]| - \lambda^*(G) = \lambda^*(G) (|S^*(G)|-1)$ many edges. This implies $S^*(G) \setminus v$ is denser than $\lambda^*(G)$ which contradicts the definition of $\lambda^*(G)$.

Since $\kappa \leq \lambda^*$, every vertex of $S^*(G)$ has a degree greater than or equal to $\kappa$. Therefore, by the uniqueness and maximality of the $\kappa$-core, we obtain $S^*(G) \subseteq H$.
    
\end{proof}

\section{Techniques and New Ideas}

At a high level, our algorithms use flow computations in a different manner from the standard max-flow formulations of densest subgraph. We do not attempt to solve a flow instance to directly find the densest subgraph.  Instead, a small number of blocking-flow phases suffices to produce a fractional assignment of edges to their endpoints. This assignment either certifies that the current graph is approximately dense or identifies a set of ``sparse'' vertices which we can indefinitely delete from our graph without decreasing the maximum density too much. We then spend greater computational effort
only after restricting the problem to a geometrically smaller graph. By iteratively decreasing the size of the graph and increasing the accuracy of our procedure, we ensure a total runtime that is linear in the input size, while also keeping the maximum density approximately the same. Eventually we will be left with an approximately dense subgraph. One can think of our procedure as a way to carve the densest subgraph out of the input graph by iteratively deleting chunks of the graph that are identified to be sparse.

\paragraph{Fractional assignments and vertex loads.}

An assignment is a set of values
$\{x_{e,u},x_{e,v}\}_{e=uv\in E}$ satisfying $ x_{e,u}+x_{e,v}=1$, with $x_{e,u},x_{e,v}\geq 0$. Assuming a single unit of mass on each edge, we interpret $x_{e,u}$ as the fraction of this mass assigned to $u$.
The \emph{load} of a vertex $v$ is then defined as
\[
    \ell_x(v)
    :=
    \sum_{e\in\delta(v)} x_{e,v}.
\]

The assignment viewpoint is motivated by an interesting property of a densest subgraph.  If $S^*$ is a densest subgraph of density $\lambda^*$, then the edges of $G[S^*]$ admit a fractional assignment under which every vertex of $S^*$ receives load $\lambda^*$.  An intuitive way to see why the greedy peeling algorithm of \cite{charikar2000} provides a $\frac{1}{2}$-approximate solution is that it uses degree of a vertex as a measure of its importance. Meanwhile, the degree is not an ideal measure as every edge contributes one unit of degree to both of its endpoints, this double counting causing the coarse $\frac{1}{2}$ approximation ratio. Our algorithm   starts from the symmetric assignment $x_{e,u}=x_{e,v}=1/2$ and uses blocking-flows to move load in a finer way from vertices above a
target threshold to vertices below it in order to make the assignment more balanced.

\paragraph{The Saturation Framework and the Main Lemma}

Given a threshold $\tau \leq (1-\varepsilon) \lambda^*(G)$, we construct a flow network in which the source and sink capacities encode,
respectively, the deficits of vertices with a load below $\tau$, denoted as \emph{under-saturated} vertices, and the excesses of vertices with a load above $\tau$, denoted as \emph{over-saturated} vertices. Sending flow along a residual path in this network accounts for moving mass from an over-saturated vertex to an under-saturated vertex, without changing the loads of internal vertices.  After
$O(\varepsilon^{-1}\log K)$ blocking-flow phases, this framework identifies a set of vertices $S$, denoted as the ``potentially saturated'' vertices, that are incident to many edges while also including most of the densest subgraph inside, formally,
\[
    |\nu_G(S)|\geq\tau|S|,
    \qquad
    |S\cap S^*(G)|
      \geq\left(1-\frac1K\right)|S^*(G)|,
\]
Finally, any vertex in $S^*(G)\setminus S$ has degree less than $2\lambda^*(G)$. Intuitively, this procedure does not fail to identify the ``important'' vertices in the densest subgraph. The running time of this procedure is $O((n+m)\varepsilon^{-2}\log K)$. For more details, see \cref{lem:saturated-discovery}. When $\tau = (1-O(\varepsilon))\lambda^*(G)$, this procedure finds a subset of the graph that contains the majority of the densest subgraph and is incident to $(1-O(\varepsilon))\lambda^*(G) |S|$ many edges. Note that this does not immediately yield a dense subgraph, as the majority of incident edges in $S$ might come from $\delta(S)$.

\paragraph{Iterative carving for DSG.} The potentially saturated vertices $S$ need not induce a dense graph, since
$|\nu_G(S)|$ also counts edges in $\delta(S)$. Therefore, instead of using $S$ to directly find a dense subgraph, we use it to carve out the instance.   At iteration $i$, we run the saturation subroutine on
the current graph $G_i = (V_i, E_i)$ with threshold
$\tau=(1-O(\varepsilon))\lambda^*(G)$ and accuracy
\[
    K_i=
    \Theta\!\left(
      \frac{1}{\varepsilon^2(1-\varepsilon)^{i-1}}
    \right).
\]
If the returned set $S_i$ contains at least a
$(1-\varepsilon)^2$ fraction of $V_i$, then its incident edges certify the entire current graph $G_i$ is already a
$(1-O(\varepsilon))$-approximate densest subgraph.  Otherwise, we continue with
$\operatorname{Core}(G_i[S_i], \gamma)$ for an arbitrary $\gamma \in \left[\frac{\lambda^*(G)}{4}, \frac{\lambda^*(G)}{2}\right]$.

This step has two simultaneous effects:
\begin{align*}
     \lambda^*(G_{i+1})
       \geq\left(1-\frac{2}{K_i}\right)\lambda^*(G_i)
    \quad \text{and} \quad
    |V_{i+1}|\leq(1-\varepsilon)^2|V_i|.
\end{align*}
The first follows from the saturation guarantees as we only lose a small fraction of the densest subgraph by carving out $V_i \setminus S_i$, while the second
follows from the condition to continue.  Since
$\sum_i1/K_i=O(\varepsilon)$ we only lose  $O(\varepsilon)$ fraction of the densest subgraph due to carving throughout the algorithm. Although the number of vertices is decreasing geometrically, the number of edges might not. The core decomposition is what converts vertex shrinkage into
graph size shrinkage. 

As the graph size decreases by a factor of $(1-\varepsilon)^2$ in each iteration, we are geometrically paying a lower computational cost per iteration despite the parameters $K_i$ increasing over time. This is what ultimately balances the runtime of the algorithm and finds an approximate densest subgraph in linear time in input size.

\paragraph{A Near-Linear $\frac{1}{2}-\varepsilon$ Approximation for DalkSP} We adapt our fractional assignments and flow-based framework to the densest at-least-$k$ subgraph problem. Unlike DSG, we cannot iteratively carve the graph to find the densest at-least-$k$ subgraph as carving too much can cause the cardinality to fall under $k$. Instead, we use the assignment framework differently, by setting the threshold to a value $\tau \leq  \frac{1}{2}\lambda^*_{\geq k}(G)$, we ensure enough slack is present that the flow framework can identify the majority of ``important vertices'' of $S^*_{\geq k}(G)$. Indeed, our framework might find less than $k$ vertices but we show these vertices are dense enough so that augmenting arbitrary vertices to them until the cardinality reaches $k$ still keeps a density of at least $(\frac{1}{2}-\varepsilon) \lambda^*_{\geq k}(G)$.
\section{An $O\left(\frac{m+n}{\varepsilon^3} \log \frac{1}{\varepsilon}\right)$ Time $(1-\varepsilon)$ Approximation Algorithm for Densest Subgraph Problem}

In this section, we present the first truly linear approximation scheme for the densest subgraph problem. Our algorithm builds on a simple yet powerful idea. Instead of directly trying to find the densest subgraph, we try to identify vertices not belonging to the densest subgraph. This enables us to sequentially carve out sparse parts of the graph until we are left with a sufficiently dense graph.

More formally, the main ingredient of our algorithm is a flow-based procedure that identifies a subset of vertices that does not contain a large chunk of the densest subgraph. The runtime of this procedure is inversely proportional to the fraction of vertices in the densest subgraph that are wrongfully identified as ``sparse'' vertices.

As we carve more and more from the graph, we are permitted to run this subroutine with increasingly more accuracy as we are dealing with a smaller graph now. To ensure that reducing the number of vertices also reduces the number of edges with the same proportions, we use tools from ``Core Decomposition'' of graphs.

This procedure eventually finds a sufficiently dense subgraph without inducing any logarithmic factors in $n$ in the runtime.

Throughout this section, we assume $E[G] \neq \emptyset$. Otherwise, the maximum density is trivially zero. This ensures $\lambda^*(G) \geq \frac{1}{2}$.

\subsection{Flow-Based Saturation Framework} In this section, we describe a flow-based subroutine, for more information on blocking-flows and their properties please see \cref{app:flow}. Our framework identifies which vertices are ``important'' for the densest subgraph. More formally, we prove the following lemma.
\begin{lemma}[Main Lemma]
    \label{lem:saturated-discovery}
    Given a graph $G$ with minimum degree at least $\frac{\lambda^*(G)}{4}$, a real number $K > 1$, a real number $0 < \varepsilon \leq \frac{1}{2}$, and a parameter $\tau \leq (1-\varepsilon)\lambda^*(G)$  where $L \cdot \tau$ is an integer for some integer $L = \Theta(\frac{1}{\varepsilon})$, there exists an algorithm that runs in time $O\left( \frac{(n+m) \log K}{\varepsilon^2}  \right)$ that returns a set of vertices $S \subseteq V$ called the \textbf{potentially saturated} vertices such that,
    \begin{itemize}
        \item $|\nu(S)| \geq \tau |S|$.
        \item $\left| S \cap S^*(G) \right| \geq |S^*(G)| \left( 1-\frac{1}{K} \right)$. 
        \item  For any vertex $v \in S^*(G) \setminus S$, $\deg_{G}(v) < 2\lambda^*(G)$.
    \end{itemize}
    We denote this algorithm by $\textsc{Saturation}(G, K, \varepsilon, \tau)$.
\end{lemma}

Intuitively, when $\tau = (1-O(\varepsilon))\lambda^*(G)$, if the potentially saturated vertices correspond to a large portion of the graph, we can ensure the current graph is sufficiently dense. Otherwise, we indefinitely delete the vertices outside $S$ from $G$ and work with $G[S]$ from thereafter. By setting $K$ properly, we sufficiently decrease the size of the graph without losing too much of the densest subgraph. To prove \cref{lem:saturated-discovery}, we need to formalize our flow-based framework.

\paragraph{Flow Framework} Assume we are given $\tau$ which we call \textbf{threshold}. Define the capacitated directed graph $H$ obtained from $G$ as follows. First, add two new vertices $s, t$ to $H$. For each edge $e = uv \in G$, add a directed edge from $u$ to $v$ and from $v$ to $u$ in $H$ both with a capacity of $\frac{1}{2}$. Moreover, for each original vertex $u$, if $\mathrm{deg}_G(u) > 2 \tau$, we add an edge directed from $u$ to  $t$ in $H$ with capacity $\frac 1 2\mathrm{deg}_G(u)- \tau $. Similarly, when $\mathrm{deg}_G(u) < 2 \tau$, we add an edge directed from $s$ to $u$ in $H$ with capacity $\tau - \frac{1}{2} \mathrm{deg}_G(u)$. Let $f^{(0)}(e) = 0$ for all edges $e \in H$. We obtain $f^{(i+1)}$ from $f^{(i)}$ by adding an arbitrary blocking flow  to $f^{(i)}$. Let $D^{(i)}$ be the residual graph of $H$ with flow $f^{(i)}$.

To proceed, we need to define an assignment function over the edges.

\begin{definition}
    \label{def:assignment}
    Given a graph $G=(V,E)$, we call $\{x_{e,u}, x_{e,v}\}_{e=uv \in E}$ an \textbf{assignment function} when for all $e=uv \in E$, $x_{e, u} + x_{e, v} = 1$ and $x_{e, u}, x_{e, v} \geq 0$. For a vertex $v \in V$, the \textbf{load} of $v$ with respect to the assignment function $x$ is defined as,
    \begin{align*}
        \ell_x(v) \coloneq \sum_{e = uv \in \delta(v)} x_{e, v}
    \end{align*}
    In other words, the load of $v$ is simply the commutative amount of value assigned to $v$ by the edges adjacent to $v$.

    For subsets $S \subseteq T \subseteq V$, we define the \emph{incoming load} of $S$ from $T$ as,
    \begin{align*}
        \ell^\mathrm{inc}_x(S,T) \coloneq \sum_{u \in S, e = uv \in \delta_T(S)} x_{e, u} 
    \end{align*}
    Similarly, the \emph{outgoing load} of $S$ from $T$ is defined as,
    \begin{align*}
        \ell_x^{\mathrm{out}}(S,T) \coloneq \sum_{u \in S, e = uv \in \delta_T(S)} x_{e, v} 
    \end{align*}

    Clearly, $\ell^\mathrm{inc}_x(S,T) + \ell_x^{\mathrm{out}}(S,T)  = |\delta_T(S)|$. 
\end{definition}

The following observation is immediate about any assignment function by definition.
\begin{observation}
    \label{obs:assignment-bounds}
    Given a graph $G = (V,E)$, $S \subseteq V$, and assignment function $x$.
    \begin{align*}
         |\nu(S)| \geq |\nu(S)| - \ell_{x}^{\mathrm{out}}(S,V)=  \sum_{v \in S} \ell_x(v) = |E[S]| + \ell_{x}^{\mathrm{inc}}(S, V)  \geq |E[S]| 
    \end{align*}
\end{observation}

The residual graph induces a fractional assignment of edges in $G$ to their endpoints. 
\begin{definition}
    Given a residual graph $D$ with respect to a flow $f$ of the graph $H$ defined above, the edge $(u,v)$ has capacity
    \begin{equation*}
        x_{e,v} = \frac{1}{2}- f(u,v) + f(v,u)
    \end{equation*}
    and the edge $(v,u)$ has capacity 
    \begin{equation*}
        x_{e,u} = \frac{1}{2}- f(v,u) + f(u,v).
    \end{equation*}
    Clearly $x_{e,u} + x_{e,v}=1$ and $x_{e,u}, x_{e,v} \geq 0$. One can therefore, think of $x_{e,v}$ as the amount edge $e$ contributes to its endpoint $v$. 
\end{definition}

We denote the assignment function given by $D^{(i)}$ as $x^{(i)}$ and the load with respect to $x^{(i)}$ as $\ell^{(i)}$. Clearly, $x^{(0)}_{e,u} = \frac{1}{2}$ for all $u \in V$ and $e \in \delta(u)$.

\begin{definition}
Given a graph $G = (V,E)$ and an assignment  function $x$ on $G$, let $\ell$ be the load with respect to $x$, then, a vertex $u \in G$ is called \textbf{over-saturated} when $\ell(u) > \tau$. Similarly, $u$ is called \textbf{under-saturated} when $\ell(u) < \tau$ and \textbf{saturated} when $\ell(u) = \tau$. 
\end{definition}

Recall the directed graph $H$ and let $x^{(0)}$ be the assignment given by $D^{(0)}$. The source $s$ has a directed edge to every under-saturated vertex $u$ and the capacity of this edge is exactly the difference between the threshold $\tau$ and load of $u$ with respect to $\ell^{(0)}$. Similarly, any over-saturated $u$ has an edge to the sink $t$ and the capacity of this edge is exactly the difference between the load of $u$ with respect to $\ell^{(0)}$ with the threshold $\tau$. 

We say a vertex $v$ is (over/under) saturated at iteration $i$ of the flow framework if it is (over/under) saturated with respect to $x^{(i)}$. The following claim is easy to see from our framework, because 
we cannot send more through a vertex $u$ than the capacity of the edge $su$, (or the capacity of $ut$ if $u$ is over-saturated.) 

\begin{claim}
    \label{obs:saturation-consistancy}
    Consider iteration $i+1$ of the flow framework. Any over-saturated vertex at iteration $i+1$ is over-saturated at iteration $i$. Similarly, any under-saturated vertex at iteration $i+1$ is under-saturated at iteration $i$. In other words, during the flow framework, an over-saturated vertex will not become under-saturated or vice versa. 

    Moreover, a saturated vertex will always remain saturated in future iterations of the flow framework.
\end{claim}
\begin{proof}
    Without loss of generality we show the claim for a vertex $u$ that was initially under-saturated. Using flow conservation and the definition of assignment $x$ induced by a flow $f$,
    \begin{align*}
        \ell_x(u) = \sum_{e \in \delta(u)} x_{e,u} = \sum_{v\in N_G(u)} \left(\frac 1 2 - f(v,u) + f(u,v) \right) = \frac{\deg(u)}{2} + f(s,u) \leq \tau
    \end{align*}
    where the last inequality follows as the capacity of $(s,u)$ on $H$ is $\tau - \frac{\deg(u)}{2}$. Therefore, an under-saturated vertex can never become over-saturated. Similarly, an over-saturated vertex cannot become under-saturated.

    Moreover, no shortest $s,t$ flow path will use an edge of the form $(u, s)$ in the residual graph as that creates a cycle which contradicts the path having the shortest length. Therefore, once a vertex becomes saturated, it remains saturated in our framework.
\end{proof}

By definition of $H$, for any vertex $v \in G$, $\ell^{(0)}(v) = \frac{\deg(v)}{2}$. Therefore, any vertex $v$ with degree  $\deg(v) \geq 2\tau$ will never become under-saturated in any iteration of the framework by \cref{obs:saturation-consistancy}.

We are finally ready to carry out the proof of \cref{lem:saturated-discovery}.

\begin{proof}[Proof of \cref{lem:saturated-discovery}]
    Using the flow framework, create a directed graph $H$ with the threshold $\tau$. 

    Now, add blocking flows to $H$ until the shortest $s,t$ flow-path has hop distance at least $T+3$ where $T =\left\lceil \frac{\log K}{\varepsilon} \right\rceil + 1$. This requires at most $\lceil \frac {\log K}{\varepsilon}\rceil + 4= O\left( \frac {\log K}{\varepsilon}\right)$ many times of adding a blocking flow. Now, let $f$ be the flow assignment at this point and $D$ the residual graph of $f$. Finally, let $x$ be the assignment induced by $f$. Since at any moment $s$ is only connected to under-saturated vertices and $t$ to over-saturated vertices, there is no path of hop distance at most $T$ between an under-saturated and an over-saturated vertex with respect to $x$.

    Now, let $S$ be the set of all vertices that are saturated or over-saturated with respect to $x$. 
    First, we bound $|\nu(S)|$.
    \begin{align*}
        |\nu(S)| &\geq \sum_{u \in S} \ell_x(u) && \rhd \cref{obs:assignment-bounds}\\
        &\geq \sum_{u \in S}  \tau && \rhd \text{Definition of Saturation} \\
        &= \tau |S|
    \end{align*}
    Now, we show $S$ contains at least a $(1-\frac{1}{K})$ fraction of the vertices in $S^*(G)$. Let $B_0$ be all the under-saturated vertices in $S^*(G)$, i.e. $B_0 = S^*(G) \setminus S$. We prove the bound using contradiction, assume $|B_0| \geq \frac{1}{K} |S^*(G)|$.
    
    Let $D^* = D[S^*(G)]$ be the induced subgraph of $D$ on $S^*(G)$, define $B_i$ as all vertices in $D^*$ within hop distance at most $i$ of $B_0$. Since there is no path of hop-distance at most $T$ between an under-saturated and an over-saturated vertex in $D$, $B_{i}$ does not contain any over-saturated vertex for all $0 \leq i \leq  T$. The following claim shows that the volume of the sets $B_i$ should increase over time.

    \begin{claim}
        \label{clm:increase-size}
        For any $T \geq i \geq 0$,
        \begin{align*}
            |E[B_{i+1}]| \geq |E[B_{i}]| + \ell_{x}^{\mathrm{out}}(B_i, S^*(G))  \geq  |E[B_{i}]|   + \varepsilon \lambda^*(G) |B_i|.
        \end{align*}
    \end{claim}
    \begin{proof}
        Consider some $T \geq i \geq 0$. By our assumptions, $|B_{i}| \geq |B_0|$. Since no vertex in $B_i$ is over-saturated, by \cref{obs:assignment-bounds} we have
        \begin{equation}
        \label{eq:edge-bound}
            \begin{aligned}
            |E[B_i]| + \ell_x^{\mathrm{inc}}(B_i, S^*(G)) &\leq |E[B_i]| + \ell_x^{\mathrm{inc}}(B_i, V) \\ &= \sum_{u \in B_i} \ell_x(u) \\ &\leq \tau \cdot |B_i|\\ &\leq (1-\varepsilon) \lambda^*(G) |B_i|. 
        \end{aligned}
        \end{equation}
        Moreover, this shows $B_i$ is a strict subset of $S^*(G)$ as if $B_i = S^*(G)$, then, $$|E[B_i]| = \lambda^*(G) |B_i| >  (1-\varepsilon) \lambda^*(G) |B_i|.$$

        Meanwhile, since $S^*(G)$ is a densest subgraph of $G$, the density of $S^*(G) \setminus B_i$ is less than or equal to the density of $S^*(G)$. Therefore,
        \begin{align*}
            \lambda^*(G) \geq \frac{|E[S^*(G) \setminus B_i]|}{|S^*(G) \setminus B_i|} = \frac{|E[S^*(G)]| - |\nu_{S^*(G)} (B_i)|}{|S^*(G)| - |B_i|}  && \rhd B_i \subseteq S^*(G)
        \end{align*}
        
        This implies $\frac{|\nu_{S^*(G)} (B_i)|}{|B_i|} \geq \lambda^*(G)$. Equivalently, $|\nu_{S^*(G)} (B_i)| \geq \lambda^*(G) |B_i|$. 
        
    Combining this with \eqref{eq:edge-bound} yields,
    \begin{align*}
       \ell_x^{\mathrm{out}}(B_i, S^*(G)) &= |\delta_{S^*(G)}(B_i)| - \ell_x^{\mathrm{inc}}(B_i, S^*(G)) \\ &=
       |\nu_{S^*(G)} (B_i)| -  |E[B_i]| - \ell_x^{\mathrm{inc}}(B_i, S^*(G)) \\ &\geq \varepsilon \lambda^*(G) |B_i|,
    \end{align*}
    where the second line substitutes $|\nu_{S^*(G)}(B_i)|$ with $|\delta_{S^*(G)}(B_i)| + |E[B_i]|$. 

    All edges $e = uv$ in $S^*(G)$ with exactly one endpoint $ u \in B_i$ and $v \notin B_i$ with a non-zero value of $x_{e,v}$ by definition satisfy $f(u,v) - f(v,u) < \frac{1}{2}$. Meanwhile, this implies there is an outgoing arc in $D^*$ from $u$ to $v$. Namely, this implies $|\delta_{D^*}^{+}(B_i)| \geq \ell_{x}^{\mathrm{out}}(B_i, S^*(G))$. Therefore,
    \begin{align*}
         |E[B_{i+1}]| \geq |E[B_{i}]|  +  |\delta_{D^*}^{+}(B_i)|   &\geq |E[B_{i}]|  + \ell_{x}^{\mathrm{out}}(B_i, S^*(G)) \\
         &\geq |E[B_{i}]| + \varepsilon \lambda^*(G) |B_i|
    \end{align*}
    which concludes the proof of this claim.
    \end{proof}
    We apply \cref{clm:increase-size} to show the size of $B_i$ will exponentially increase over time.
    \begin{claim}
        \label{clm:increase-size-vertices}
        For any $T-1 \geq i \geq 0$,
        \begin{align*}
            |B_{i+1}| \geq e^{\varepsilon}|B_i|
        \end{align*}
    \end{claim}
    \begin{proof}
        Since $i \leq T-1$, $B_{i+1}$ does not contain any over-saturated vertex. Meanwhile, \cref{clm:increase-size} implies
         \begin{align*}
            \ell_{x}^{\mathrm{out}}(B_i, S^*(G)) 
         &\geq  \varepsilon \lambda^*(G) |B_i|
        \end{align*}
        Consider the edge $e =uv$ where $u \in B_i$ and $v \notin B_i$, $e$ contributes $x_{e,v}>0$ to the load of $v$ as the directed edge $(u,v)$ is in $D^*$. Moreover, $v \in B_{i+1} \setminus B_i$ as there is a directed edge from $u$ to $v$ in $D^*$, this implies,
        \begin{align}
            \label{eq:marg-lowerbound}
            \sum_{u \in B_{i+1} \setminus B_i} \ell_x(u) \geq \ell_{x}^{\mathrm{out}}(B_i, S^*(G)) 
        \end{align}
        
        Meanwhile, since no vertex in $B_{i+1}$ is over-saturated, 
        \begin{align}
            \label{eq:marg-upperbound}
            \sum_{u \in B_{i+1} \setminus B_i} \ell_x(u) \leq  |B_{i+1} \setminus B_i| \tau \leq (1-\varepsilon)\lambda^*(G) |B_{i+1} \setminus B_i|. 
        \end{align}
        Combining \eqref{eq:marg-lowerbound} and \eqref{eq:marg-upperbound} yields,
        \begin{align*}
            \varepsilon \lambda^*(G) |B_i|\leq (1-\varepsilon)\lambda^*(G) |B_{i+1} \setminus B_i|.
        \end{align*}
        Finally, it follows that,
        \begin{align*}
           |B_{i+1}| = |B_{i+1} \setminus B_i| + |B_i| \geq \frac{|B_i|}{1-\varepsilon} \geq e^\varepsilon |B_i|.
        \end{align*}
        where the last inequality uses $\frac{1}{1-x} \geq e^x \:, \forall x\in[0,1)$. This concludes the proof.
    \end{proof}
    
    Returning to the proof of \cref{lem:saturated-discovery}, note that iteratively applying \cref{clm:increase-size-vertices} yields,
    \begin{align*}
        \left|B_{T}\right| \geq e^{T\varepsilon}|B_0| > K |B_0| && \rhd T = \left\lceil \frac{\log K}{\varepsilon} \right\rceil + 1.
    \end{align*}
    Recall that we are assuming $|B_0| \geq \frac{|S^*(G)|}{K}$. This  implies $\left|B_{T}\right| > |S^*(G)|$.
    Since $B_T \subseteq S^*(G)$ it follows that $|B_T| \leq |S^*(G)|$, which yields a contradiction.
    It follows that $|B_0| < \frac{1}{K} |S^*(G)|$ or equivalently, $\left| S \cap S^*(G) \right| \geq |S^*(G)| \left( 1-\frac{1}{K} \right)$.

    Finally, by \cref{obs:saturation-consistancy}, a vertex that is not under-saturated will never become under-saturated during the algorithm. By our initial assignment, any vertex $v$ with $\deg(v) \geq 2\lambda^*$ has load $\frac{\deg(v)}{2} \geq \lambda^* \geq \tau$ in the beginning of the algorithm. Thus, such a vertex will never become under-saturated. This implies any vertex $v \in S^*(G) \setminus S$ satisfies $\deg(v) < 2\lambda^*(G)$.

    \paragraph{Runtime of the Procedure.} To find the set $S$, we simply need to compute at most $T+2 = O\left(  \frac{\log K}{\varepsilon}\right)$ many blocking flows. Recall we are assuming $L \tau$ is an integer for some $L = \Theta\left(\frac{1}{\varepsilon} \right)$. This implies the capacity on every edge of the directed graph $H$ is divisible by $\frac{1}{2L} = \Theta(\varepsilon)$. Moreover, it is easy to see that the summation of the capacity of the edges is at most $2m+n\lambda^*(G)$. Since $G$ is a $\frac{\lambda^*(G)}{4}$-core, $m \geq \frac{\lambda^*(G) n}{8}$. Thus, the summation of the capacities in the network $H$ is $O(n+m)$. Consequently, using a depth-first search, we can find a blocking flow for $H$ in time $O\left( \frac{m+n}{\varepsilon} \right)$. Since we need to find $O\left(  \frac{\log K}{\varepsilon}\right)$ many blocking flows, the algorithm takes $O\left(  \frac{(m+n)\log K}{\varepsilon^2}\right)$ time.
\end{proof}

\subsection{Carving the Densest Subgraph: Combining the Flow Framework and Graph Cores} This section presents our main algorithm, a linear time approximation scheme for the densest subgraph problem.

\thmMainUndir*

\paragraph{Description of the Algorithm} The algorithm runs in multiple iterations, each iteration deals with a geometrically smaller graph. At iteration $i$, the algorithm deals with a graph $G_i = (V_i, E_i)$ which is a subgraph of the original graph $G$. Let $n_i, m_i$ respectively denote the number of vertices and edges in $G_i$. We set the threshold parameter $\tau$ to some value in $[(1-\varepsilon)^3 \lambda^*(G), (1-\varepsilon)^2\lambda^*(G)]$ such that there exists some integer $L = \Theta\left( \frac{1}{\varepsilon} \right)$ such that $L\tau$ is an integer.

\begin{itemize}
    \item \textbf{Initialization:} Let $\gamma \in [\lambda^*(G)/4, \lambda^*(G)/2]$. Such $\gamma$ can be obtained in linear time by running the $2$-approximation algorithm  of \cite{charikar2000} and halving the density of the densest subgraph returned by that algorithm. Let $G_1 = \textsc{Core}(G, \gamma)$.
    \item \textbf{Iteration $i$:} At iteration $i$, the algorithm runs $\textsc{Saturation}(G_i, K_i, \varepsilon, \tau)$ where $K_i = \frac{4}{\varepsilon^2(1-\varepsilon)^{i-1}}$ to obtain a set $S_i$. Depending on the cardinality of $S_i$, the algorithm acts differently. In every run of the $\textsc{Saturation}$ subroutine, we use the same threshold parameter $\tau$, later we explain how such $\tau$ can be found.
    \begin{itemize}
        \item \textbf{If $|S_i| \geq (1-\varepsilon)^2 n_i$:} The algorithm returns $V_i$ as an approximate densest subgraph of $G$.
        \item \textbf{If $|S_i| < (1-\varepsilon)^2 n_i$:} The algorithm sets $G_{i+1} = \textsc{Core}(G_i[S_i], \gamma)$ and proceeds to the next iteration.
    \end{itemize}

\end{itemize}

\begin{algorithm}[t]
\caption{A linear approximation scheme for the densest subgraph problem.}
\label{alg:linear}
\begin{algorithmic}[1]
\State {\bfseries Input:} $G=(V,E), \tau, \varepsilon$
    \State Obtain $\gamma$ by running Charikar's 2-approximation algorithm \cite{charikar2000} and halving the output
    \State $G_1 = (V_1, E_1) \gets \textsc{Core}(G, \gamma)$
    \State $n_1 \gets |V_1|, m_1 \gets |E_1|$
    \State $i \gets 1$
    \While{{True}}
        \State $K_i \gets \frac{4}{\varepsilon^2 (1-\varepsilon)^{i-1}}$
        \State $S_i \gets \textsc{Saturation}(G_i, K_i, \varepsilon, \tau)$
        \If{$|S_i| \geq (1-\varepsilon)^2 n_i$}
            \State \Return $V_i$
        \Else
            \State $G_{i+1} = (V_{i+1}, E_{i+1}) \gets \textsc{Core}(G_i[S_i], \gamma)$
            \State  $n_{i+1} \gets |V_{i+1}|, m_{i+1} \gets |E_{i+1}|$
            \State $i \gets i+1$
        \EndIf
    \EndWhile
\end{algorithmic}
\end{algorithm}

We start by showing several invariants of our algorithm.

\subsubsection{Structural Invariants}
We start by investigating $\lambda^*(G_i)$ for the graph $G_i$ considered in iteration $i$ of our algorithm.  In particular,  \cref{lem:densitybound-iteration} and \cref{clm:size-bound} below together show that in each iteration the size of the instance $G_i$ decreases while the maximum density approximately stays the same.

The following lemma establishes a lower-bound on $\lambda^*(G_i)$.

\begin{lemma}
\label{lem:densitybound-iteration}
    If \cref{alg:linear} runs for at least $i$ iterations, then, 
    \begin{align*}
        \lambda^*(G_i) \geq \left(\prod_{j=1}^{i-1} \left( 1- \frac{2}{K_j} \right)\right) \lambda^*(G) \geq (1-\varepsilon) \lambda^*(G)
    \end{align*}
    Moreover, this implies the threshold $\tau$ satisfies,
    \begin{align*}
        \tau \leq (1-\varepsilon) \lambda^*(G_i)
    \end{align*}
\end{lemma}
\begin{proof}
    We show the lemma holds by induction. Firstly, the base case holds as $\lambda^*(G_1) =\lambda^*(G)$. Indeed, any vertex $v \in S^*(G)$ has at least $\lambda^*$ many edges to other vertices in $S^*(G)$, otherwise, $S^*(G) \setminus \{v\}$ has a strictly larger density. Since $\gamma \leq \frac{\lambda^*}{2}$, all vertices in $S^*(G)$ are still in $G_1 = \textsc{Core}(G, \gamma)$.

    Now, assume the claim holds for graph $G_i$.  We show it holds for $G_{i+1}$. By \cref{lem:saturated-discovery}, $|S^*(G_i) \cap S_i| \geq (1-\frac{1}{K_i})|S^*(G_i)|$. Moreover, any vertex in $S^*(G_i) \setminus S_i$ has at most $2\lambda^*(G_i) \leq 2\lambda^*(G)$ many edges to $S^*(G_i)$. Now, consider the set $S^*(G_i) \cap S_i$. This set has density
    \begin{align*}
        \frac{|E[S^*(G_i) \cap S_i]|}{|S^*(G_i) \cap S_i|} &\geq \frac{|E[S^*(G_i)]| - \sum_{v \in S^*(G_i)\setminus S_i} \deg_{G_i}(v)}{|S^*(G_i)|} \\  &> \frac{|E[S^*(G_i)]| - 2\lambda^*(G_i)|S^*(G_i) \setminus S_i|}{|S^*(G_i)|} && \rhd \forall v \in S^*(G_i) \setminus S_i \: \deg_{G_i}(v) < 2\lambda^*(G_i) \\
        &> \frac{|E[S^*(G_i)]| - \lambda^*(G_i)\frac{2}{K_i}|S^*(G_i)|}{|S^*(G_i)|} && \rhd |S^*(G_i) \setminus S_i| < \frac{1}{K_i}|S^*(G_i)|\\
        &= \left(1- \frac{2}{K_i}\right) \lambda^*(G_i) && \rhd |E[S^*(G_i)]| = \lambda^*(G_i) |S^*(G_i)|\\
        &\geq \left(\prod_{j=1}^{i} \left( 1- \frac{2}{K_j} \right)\right) \lambda^*(G) && \rhd \text{Inductive Hypothesis}
    \end{align*}

    Since $S^*(G_i) \cap S_i \subseteq G_{i}[S_i]$, this imposes the desired bound on the maximum density of $G_{i}[S_i]$.

    Moreover, by the definition of $K_j$ for any $j$, we have,
    \begin{align*}
        \lambda^*(G_{i}[S_i]) &\geq \left(\prod_{j=1}^{i} \left( 1- \frac{2}{K_j} \right)\right) \lambda^*(G)\\ &= \left(\prod_{j=1}^{i} \left( 1- \frac{\varepsilon^2(1-\varepsilon)^{j-1}}{2} \right)\right) \lambda^*(G) && \rhd K_j = \frac{4}{\varepsilon^2 (1-\varepsilon)^{j-1}}\\
        &\geq \left(\prod_{j=1}^{i} e^{- \varepsilon^2(1-\varepsilon)^{j-1}} \right) \lambda^*(G) && \rhd (1-x) \geq e^{-2x} \:, \forall x \in [0, 0.5] \\
        &= e^{-\varepsilon^2 (\sum_{j=1}^i (1-\varepsilon)^{j-1})} \lambda^*(G) \\
        &\geq e^{- \varepsilon} \lambda^*(G) && \rhd \sum_{j = 0}^{\infty} (1-\varepsilon)^j = \frac{1}{\varepsilon} \\
        &\geq (1-\varepsilon) \lambda^*(G) && \rhd (1-x) \leq e^{-x} \: , \forall x \in \mathbb R
    \end{align*}
    Now, since $(1-\varepsilon) \lambda^*(G) \geq \frac{\lambda^*(G)}{2}$ it follows that  $\lambda^*(G_i[S_i]) \geq \frac{\lambda^*(G)}{2}$. Thus, all vertices of $S^*(G_i[S_i])$ are in the $\gamma-$core of $G_i[S_i]$. Since $G_{i+1} = \textsc{Core}(G_i[S_i], \gamma)$, we have that 
    \begin{align*}
        \lambda^*(G_{i+1}) = \lambda^*(G_i[S_i]) \geq \left(\prod_{j=1}^{i} \left( 1- \frac{2}{K_j} \right)\right) \lambda^*(G) \geq (1-\varepsilon) \lambda^*(G).
    \end{align*}

    Finally, since $\tau \leq (1-\varepsilon)^2 \lambda^*(G)$ and $\lambda^*(G_i) \geq (1-\varepsilon) \lambda^*(G)$, we must have 
    \begin{align*}
        \tau \leq (1-\varepsilon)\lambda^*(G_i)
    \end{align*}
    which concludes the induction step and  the proof of \cref{lem:densitybound-iteration}.

\end{proof}

Note that \cref{lem:densitybound-iteration} shows $\tau$ always satisfies the upper bound required in the $\textsc{Saturation}$ subroutine. Moreover, since each $G_i$ is a $\gamma$-core, the minimum degree condition is also satisfied.

We now show  that the size of the graph $G_i$ decreases adequately in each iteration.

\begin{claim}
    \label{clm:size-bound}
    If \cref{alg:linear} runs for at least $i$ iterations, then,
    \begin{align*}
        n_i \leq (1-\varepsilon)^{2(i-1)} n_1 \quad \text{and} \quad m_i \leq 8(1-\varepsilon)^{2(i-1)} m_1
    \end{align*}
\end{claim}

\begin{proof}
    We prove the claim by induction. The base case trivially holds. Now, assume the upper bound holds for $n_i, m_i$, we investigate $n_{i+1}$ and $m_{i+1}$. Since the algorithm has not terminated in iteration $i$, we have $|S_i| < (1-\varepsilon)^2 n_i$, thus, $\textsc{Core}(G_i[S_i], \gamma)$ has less than $(1-\varepsilon)^2 n_i$ many vertices. Using the inductive hypothesis we get, 
    \begin{align*}
        n_{i+1} \leq (1-\varepsilon)^2 n_i \leq (1-\varepsilon)^{2i} n_1
    \end{align*}

    Now, consider $m_{i+1}$, since the density of $V_{i+1}$ can be at most $\lambda^*(G)$, we get $m_{i+1}\leq \lambda^*(G) n_{i+1}$. Using the upper bound on $n_{i+1}$ we get,
    \begin{align*}
        m_{i+1} \leq \lambda^*(G) (1-\varepsilon)^{2i} n_1
    \end{align*}
    
    Moreover, since $G_{1}$ is the $\gamma$-core of a graph, it follows that 
    \begin{align*}
        m_1 \geq \frac{\gamma}{2} n_1 \geq \frac{\lambda^*(G)}{8} n_1,
    \end{align*}
    where we used $\gamma \in [\lambda^*(G)/4, \lambda^*(G)/2]$. Putting the two inequalities together yields,
    \begin{align*}
        m_{i+1} \leq 8(1-\varepsilon)^{2i}  m_1.
    \end{align*}
\end{proof}

\begin{remark}
    \cref{clm:size-bound} implies \cref{alg:linear} terminates after $O\left( \frac{\log n}{\varepsilon} \right)$ iterations. Despite this, we show the amortized runtime of the algorithm avoids any logarithmic factor.
\end{remark}

\subsubsection{Proof of \cref{thm:undir}}
\begin{proof}[Proof of \cref{thm:undir}]
    We are finally ready to proceed with the proof of \cref{thm:undir}. We start by showing the approximation ratio guarantee.

\begin{claim}
    If \cref{alg:linear} terminates at iteration $i$, then $V_i \subseteq V$ has density $(1-5\varepsilon) \lambda^*(G)$. 
\end{claim}
\begin{proof}
    If the algorithm terminates at iteration $i$, we must have $|S_i| \geq (1-\varepsilon)^2 n_i$. Moreover, by \cref{lem:saturated-discovery}, $|\nu_{G_i}(S_i)| \geq \tau|S_i| \geq  (1-\varepsilon)^3 \lambda^*(G) |S_i|$. Putting the two together gives,
    \begin{align*}
        m_i \geq |\nu_{G_i}(S_i)| &\geq (1-\varepsilon)^3 \lambda^*(G) |S_i| \\
        &\geq (1-\varepsilon)^5 \lambda^*(G) n_i && \rhd |S_i| \geq (1-\varepsilon)^2 n_i
    \end{align*}
    Thus, we obtain the following bound on the density of $V_i$,
    \begin{align*}
        \frac{m_i}{n_i} \geq (1-\varepsilon)^5 \lambda^*(G) \geq (1-5\varepsilon) \lambda^*(G)
    \end{align*}
    where we used $(1-x)^5 \geq 1-5x \:, \forall x \in [0, 1]$.
\end{proof}

Now, we establish the runtime of this procedure,
\begin{claim}
    The runtime of \cref{alg:linear} is bounded by $O\left( \frac{n+m}{\varepsilon^3}\log \frac{1}{\varepsilon}\right)$.
\end{claim}
\begin{proof}
By \cref{lem:core}, computing the graph $G_1$ takes $O(n+m) $ time.  Assume the algorithm runs for $T$ iterations. For each $1 \leq i \leq T$, computing the set $S_i$ takes $O\left( \frac{(n_i+m_i) \log K_i}{\varepsilon^2} \right)$. After that, deciding the appropriate case and computing the $\gamma$-core of $G_i[S_i]$, if needed, takes $O(n_i + m_i)$ time. Thus, the overall runtime of the algorithm is bounded by,
\begin{align*}
    O\left( n+m + \sum_{i=1}^T \frac{(n_i+m_i) \log K_i}{\varepsilon^2}  \right) 
\end{align*}
We can further bound this term by the following:
\begin{align*}
    n+m + \sum_{i=1}^T \frac{(n_i+m_i)\log K_i}{\varepsilon^2} & \leq n+m + 8\sum_{i=1}^T \frac{(n+m) (1-\varepsilon)^{2(i-1)} \log K_i}{\varepsilon^2} \\
    &\leq n+m+\frac{8(n+m)}{\varepsilon^2} \sum_{i=1}^\infty (1-\varepsilon)^{2(i-1)} \left( 2\log \frac{1}{\varepsilon} + \log \frac{4}{(1-\varepsilon)^{i-1}}  \right)
\end{align*}
where the first line uses \cref{clm:size-bound}  and the second line uses $K_i = \frac{4}{\varepsilon^2(1-\varepsilon)^{i-1}}$.

Now, using $\log x \leq x$ for $x \geq 1$ we obtain,
\begin{align*}
    \sum_{i=1}^\infty (1-\varepsilon)^{2(i-1)} \left( 2\log \frac{1}{\varepsilon} + \log \frac{4}{(1-\varepsilon)^{i-1}} \right) &\leq 2\log \frac{1}{\varepsilon}\sum_{i=1}^\infty (1-\varepsilon)^{i-1} + 4\sum_{i=1}^\infty \frac{(1-\varepsilon)^{2(i-1)}}{(1-\varepsilon)^{i-1}} \\
    &\leq \left( 2\log \frac{1}{\varepsilon} + 4 \right) \sum_{i=1}^\infty (1-\varepsilon)^{i-1} \\
    &= O\left( \frac{1}{\varepsilon} \log \frac{1}{\varepsilon} \right)
\end{align*}
\end{proof}

Thus, given a feasible $\tau$, we can bound the overall runtime as
\begin{align*}
     n+m+\frac{8(n+m)}{\varepsilon^2} \sum_{i=1}^\infty (1-\varepsilon)^{2(i-1)} \left( 2\log \frac{1}{\varepsilon} + \log \frac{1}{(1-\varepsilon)^{i-1}} \right) =O \left( \frac{n+m}{\varepsilon^3} \log \frac{1}{\varepsilon} \right)
\end{align*}

\paragraph{Finding the Desired Threshold Parameter $\tau$ Efficiently.} Finally, we explain how we can obtain the desired threshold parameter $\tau$ with a constant overhead. Simply using a binary search to find $\tau$ adds an overhead of $\log \frac{1}{\varepsilon}$. To avoid this, we use the following simple trick.

Let $\varepsilon_i = 2^{-(i+2)}$ for all $i \in \mathbb N_{\geq 1}$. Firstly, we obtain a $1-5\varepsilon_1$-approximation for $\lambda^*(G)$ in $O(n+m)$ by running Charikar's greedy peeling algorithm \cite{charikar2000}. Note that $1-5\varepsilon_1 = \frac{3}{8} < \frac{1}{2}$. Let $\lambda_1$ denote this value. Given a $(1-5\varepsilon_i)$-approximation solution with density $\lambda_i$, we know $\lambda^*(G) \in [\lambda_i, \frac 1 {(1-5\varepsilon_i)} \lambda_i]$, thus, using constantly many calls to \cref{alg:linear} with parameter $\varepsilon_{i+1}$ for different values of $\tau \in \left[(1-\varepsilon_{i+1})^3\lambda_i, \frac {(1-\varepsilon_{i+1})} {(1-5\varepsilon_i)} \lambda_i\right]$ we obtain a $(1-5\varepsilon_{i+1})$ approximation solution with an additional runtime of $O\left( \frac{n+m}{\varepsilon_{i+1}^3} \log \frac{1}{\varepsilon_{i+1}} \right)$.

To ensure that there exists a large integer $L = \Theta\left(\frac{1}{\varepsilon} \right)$, simply let $L = \lceil \frac{c}{\varepsilon}\rceil$ for a large enough constant $c$ and simply let $\tau \leftarrow \frac{1}{L}\lfloor L \tau \rfloor$, this is feasible as $\lambda^*(G) \geq \frac{1}{2}$.

We repeat this until we obtain a $(1-5\varepsilon_I)$ approximation where $5\varepsilon_I \leq \varepsilon < 5\varepsilon_{I-1}$. The overall runtime of this procedure is dominated by the last iteration of it, since $\varepsilon < 5\varepsilon_{I-1} = 10 \varepsilon_I$, the overall runtime of the algorithm is bounded by $O\left( \frac{n+m}{\varepsilon^3} \log \frac{1}{\varepsilon} \right)$.
\end{proof}

\section{A Near-Linear $(\frac{1}{2}-\varepsilon)$-Approximation for the Densest At-Least-$k$ Subgraph Problem}

In this section, we prove \cref{thm:atleast} by adapting our flow framework to the densest at-least-$k$ subgraph problem. We assume $E[G] \neq \emptyset$ as otherwise, the densest at-least-$k$ subgraph has zero density. This ensures $\lambda^*_{\geq k}(G) \geq \frac{1}{n}$.

\begin{proof}[Proof of \cref{thm:atleast}]

Assume we are given $\tau \leq  \frac{ \lambda^*_{\geq k}(G)}{2}$. Using the flow framework, create a directed graph $H$ with threshold $\tau$. Now, add blocking flows to $H$ until the shortest $s,t$ flow-path has hop distance at least $T+3$ for $T = \left\lceil \frac{\log m}{\varepsilon} \right\rceil + 2$. 

Let $D$ denote the residual graph of $H$ at this moment. We utilize the orientation given by $D$ to construct an approximately dense subgraph with at least $k$ vertices. Let $x$ be the assignment function induced by $D$ and $\ell_x$ the load corresponding to it.

Obtain $D'$ from $D$ by deleting $s, t$ and reversing the direction of all edges.

Now, let $B_0$ be the set of all over-saturated vertices in $D$ and let $B_i$ be the set of all vertices in $D'$ of hop distance at most $i$ from $B_0$. Since there are no $s$-$t$ flow paths of length at most $T+2$, there is no path of length at most $T$ from an over-saturated vertex to an under-saturated vertex in $D'$. I.e., $B_T$ does not contain any under-saturated vertex. We divide the analysis into two cases. 

\begin{itemize}
    \item \textbf{Case 1: $|B_1| < k$.} We prove $ \sum_{u \in B_0} \ell_x(u) \geq \frac{k}{2}\lambda^*_{\geq k}(G)$. Consider the densest at-least-$k$ subgraph $S^*_{\geq k}(G)$, by \cref{obs:assignment-bounds},
    \begin{align}
        \label{eq:atleast-lower}
        \sum_{u \in S^*_{\geq k}(G)} \ell_x(u) \geq |E[S^*_{\geq k}(G)]| = \lambda^*_{\geq k}(G)|S^*_{\geq k}(G)| 
    \end{align}

    Now, let $U^*$ be the non over-saturated vertices in $S^*_{\geq k}(G)$. By definition,
    \begin{align}
        \label{eq:atleast-upper}
         \sum_{u \in U^*} \ell_x(u) \leq \tau |U^*| \leq \tau  |S^*_{\geq k}(G)| 
    \end{align}

    Therefore,  $B^* = B_0\cap S^*_{\geq k}(G)$, the over-saturated vertices in $S^*_{\geq k}(G)$, must satisfy

    \begin{align*}
        \sum_{u \in B^*} \ell_x(u) &= \left( \sum_{u \in S^*_{\geq k}(G)}\ell_x(u) \right) - \left(  \sum_{u \in U^*}\ell_x(u) \right) \\
        &\geq (\lambda^*_{\geq k}(G) - \tau) |S^*_{\geq k}(G)| \\
        &\geq \frac{\lambda^*_{\geq k}(G)}{2} |S^*_{\geq k}(G)| 
    \end{align*}
    where the second line combines  \eqref{eq:atleast-lower} and \eqref{eq:atleast-upper} and the third line uses $\tau \leq \frac{\lambda^*_{\geq k}(G)}{2}$.

    Since $B^* \subseteq B_0$, the load on the vertices in $B^*$ is equal to $|E[B^*]| + \ell_{x}^{\mathrm{inc}}(B^*, B_1)$. Since the edges $e=uv$ where $u \in B^*, v \notin B^*$ that contribute a positive $x_{e, u}$ to $\ell_{x}^{\mathrm{inc}}(B^*, B_1)$ are directed in the residual graph from $v$ to $u$, they are present in $E[B_1]$.
    
    Therefore, by \cref{obs:assignment-bounds},
    \begin{equation}
    \label{eq:B1-bound}
    \begin{aligned}
        |E[B_1]|
        &\geq |E[B^*]| + \ell_{x}^{\mathrm{inc}}(B^*, B_1) \\
        &= \sum_{u \in B^*} \ell_x(u) \\
        &\geq \frac{\lambda^*_{\geq k}(G)}{2} k
    \end{aligned}
    \end{equation}

    Since $|B_1| < k$, let $B^\dagger$ be obtained by arbitrarily adding $k-|B_1|$ new vertices to $B_1$. Now, the density of $B^\dagger$ is at least,
    \begin{align*}
        \frac{|E[B^\dagger]|}{|B^\dagger|} \geq \frac{|E[B_1]|}{k} \geq \frac{\lambda^*_{\geq k}}{2}
    \end{align*}
    Therefore, in this case we obtain a $\frac{1}{2}$ approximation solution for the densest at-least-$k$ subgraph problem by arbitrarily adding vertices to $B_1$ until it has size $k$.
    
    \item \textbf{Case 2: $|B_1| \geq k$.} Since there are no augmenting flow paths of length at most $T+2$, any path from an over-saturated vertex to an under-saturated vertex has length at least $T$ in $D'$. Therefore, for all $0\leq i \leq T$, $B_i$ does not contain under-saturated vertices.  We conclude the proof in this case by proving the following lemma.
    \begin{lemma}
    \label{lem:at-least-dense-discover}
    Assume there are no $s$-$t$ augmenting paths of length at most $T+2$ in $D$ and $|B_1| \geq k$. Then, for some $1 \leq i \leq T$, $B_i$ has density at least $(1-\varepsilon) \tau$ and $|B_i| \geq k$.
    \end{lemma}
    
    \begin{proof} Assume for contradiction that no $B_i$ has density $(1-\varepsilon) \tau$.
        For $0\leq i \le T$, $B_i$ does not contain any under-saturated vertex. Moreover, by definition, an edge $e=uv$ with $u \in B_i$, $v \notin B_i$ contributes $x_{e,u}$ to $\ell_{x}^{\mathrm{inc}}(B_i, B_{i+1})$ if and only if the directed edge $(v,u)$ is present in the residual graph $D$. Thus, all edges that contribute $x_{e,u} > 0$ to $\ell_{x}^{\mathrm{inc}}(B_i, B_{i+1})$ must be inside $B_{i+1}$ by definition. This implies,
    \begin{align*}
         |E[B_{i+1}]| &\geq |E[B_i]| + \ell_{x}^{\mathrm{inc}}(B_i, B_{i+1}) \\ &= \sum_{u \in B_i}\ell_x(u)
         \\ &\geq   \sum_{u \in B_i} 
        \tau  \\
    &= \tau |B_i|
    \end{align*}
    where the second line follows from \cref{obs:assignment-bounds} and the last line follows since no vertex in $B_i$ is under-saturated.
    
    Moreover, since $B_i$ is not $(1-\varepsilon) \tau$ dense,
    \begin{align*}
          |E[B_i]| < (1-\varepsilon)  \tau |B_i|
    \end{align*}
    This implies,
    \begin{align}
        \label{eq:edgecountbound-atleast}
        |E[B_{i+1}]| \geq \frac{1}{1-\varepsilon} |E[B_i]| \geq e^\varepsilon |E[B_i]| 
    \end{align}
    where we used $\frac{1}{1-x} \geq e^x \;, \forall x\in [0, 1)$.

    In \eqref{eq:B1-bound}, we showed $|E[B_1]| \geq \frac{\lambda^*_{\geq k}(G)}{2} k > 0$, since the number of edges is an integer, $|E[B_1]|\geq 1$, iteratively applying \eqref{eq:edgecountbound-atleast} gives, $|E[B_{T}]| \geq e^{\varepsilon(T-1)} > m$ as $T = \left\lceil \frac{\log m}{\varepsilon} \right\rceil + 2$. This completes the contradiction as $B_T$ is a subgraph of $G$.
    
    Therefore, some $B_i$ must have density at least $(1-\varepsilon)\tau$.
    \end{proof}
    By \cref{lem:at-least-dense-discover}, we can simply check all $B_i$ in time $O(n+m)$ and find a subgraph with density $(1-\varepsilon) \tau$. Note that to obtain this graph we needed to add $T = O\left( \frac{\log n}{\varepsilon} \right)$ many blocking flows to $D$. Each blocking flow can be found in $O\left( {(m+n) \log n} \right)$ using dynamic trees \cite{Goldberg1990FindingMC,SLEATOR1983362}, therefore, we can find a $(1-\varepsilon)\tau$ dense solution in time $O\left( \frac{(m+n)\log^2 n}{\varepsilon} \right)$.
    
    Now, we utilize this to obtain  a $\frac{1}{2}-\varepsilon$ approximation solution. Assume we are given a $\tau \in [(\frac{1}{2}- \varepsilon') \lambda^*_{\geq k}(G), \frac{1}{2} \lambda^*_{\geq k}(G)]$, using the above procedure we can turn this into a $(1-\varepsilon')\tau \geq (1-\varepsilon')(\frac{1}{2}- \varepsilon') \lambda^*_{\geq k}(G)$ dense solution. Since $(\frac{1}{2} - x)(1-x) \geq (\frac{1}{2}-2x)$ for $x \in \mathbb R_{\geq 0}$, this solution is a $\frac{1}{2}-2\varepsilon'$ approximation solution for the densest-at-least-$k$ subgraph problem.

    Firstly, we obtain a $\frac1 3$-approximation for the densest-at-least-$k$ subgraph problem by running Andersen's linear time algorithm \cite{Andersen2009}, let $\lambda$ denote this value. We know  $\lambda^*_{\geq k}(G) \in [\lambda, 3 \lambda]$, thus, using $O\left(\log \frac{1}{\varepsilon} \right)$ calls to our procedure with a binary search we can obtain a $\tau \in [\frac{1-\varepsilon}{2} \lambda^*_{\geq k}(G), \frac{1}{2} \lambda^*_{\geq k}(G)]$. 

Using such $\tau$ gives a $(\frac{1}{2} - \frac{\varepsilon}{2}) (1-\frac{\varepsilon}{2}) \geq \frac{1}{2} - \varepsilon$-approximation solution to the densest at-least-$k$ subgraph problem. The overall runtime of the algorithm is bounded by $O\left( \frac{(n+m)\log^2 n \log \frac{1}{\varepsilon}}{\varepsilon}  \right)$.

\end{itemize}

\end{proof}

\paragraph{Acknowledgments} We are grateful to Kent Quanrud for the invaluable insights he shared with us during the early stages of this project.

\bibliographystyle{plain}
\bibliography{references}

\begin{thebibliography}{10}

\bibitem{andersen2010local}
Reid Andersen.
\newblock A local algorithm for finding dense subgraphs.
\newblock {\em ACM Transactions on Algorithms (TALG)}, 6(4):1--12, 2010.

\bibitem{Andersen2009}
Reid Andersen and Kumar Chellapilla.
\newblock Finding dense subgraphs with size bounds.
\newblock In {\em Algorithms and Models for the Web-Graph}, pages 25--37. Springer Berlin Heidelberg, 2009.

\bibitem{AsahiroHassinIwama02}
Yuichi Asahiro, Refael Hassin, and Kazuo Iwama.
\newblock Complexity of finding dense subgraphs.
\newblock {\em Discrete Applied Mathematics}, 121(1--3):15--26, 2002.

\bibitem{asahiro1996}
Yuichi Asahiro, Kazuo Iwama, Hisao Tamaki, and Takeshi Tokuyama.
\newblock Greedily finding a dense subgraph.
\newblock In {\em Algorithm Theory (SWAT)}, pages 136--148, Berlin, Heidelberg, 1996. Springer Berlin Heidelberg.

\bibitem{bahmani2014}
Bahman Bahmani, Ashish Goel, and Kamesh Munagala.
\newblock Efficient primal-dual graph algorithms for {MapReduce}.
\newblock In Anthony Bonato, Fan~Chung Graham, and Pawe{\l} Pra{\l}at, editors, {\em Algorithms and Models for the Web Graph}, volume 8882 of {\em Lecture Notes in Computer Science}, pages 59--78. Springer, 2014.

\bibitem{batagelj2003cores}
Vladimir Batagelj and Matja{\v z} Zaver{\v s}nik.
\newblock An {$O(m)$} algorithm for cores decomposition of networks.
\newblock {\em CoRR}, cs.DS/0310049, 2003.

\bibitem{BhaskaraEtAl10}
Aditya Bhaskara, Moses Charikar, Eden Chlamtac, Uriel Feige, and Aravindan Vijayaraghavan.
\newblock Detecting high log-densities: An $o(n^{1/4})$ approximation for densest $k$-subgraph.
\newblock In {\em Proceedings of the 42nd ACM Symposium on Theory of Computing}, pages 201--210. Association for Computing Machinery, 2010.

\bibitem{boob2020}
Digvijay Boob, Yu~Gao, Richard Peng, Saurabh Sawlani, Charalampos~E. Tsourakakis, Di~Wang, and Junxing Wang.
\newblock Flowless: Extracting densest subgraphs without flow computations.
\newblock In {\em Proceedings of The Web Conference 2020}, WWW '20, pages 573--583. Association for Computing Machinery, 2020.

\bibitem{BravermanEtAl17}
Mark Braverman, Young~Kun Ko, Aviad Rubinstein, and Omri Weinstein.
\newblock {ETH} hardness for densest-$k$-subgraph with perfect completeness.
\newblock In {\em Proceedings of the 2017 Annual ACM-SIAM Symposium on Discrete Algorithms}, pages 1326--1341. Society for Industrial and Applied Mathematics, 2017.

\bibitem{charikar2000}
Moses Charikar.
\newblock Greedy approximation algorithms for finding dense components in a graph.
\newblock In {\em Approximation Algorithms for Combinatorial Optimization}, volume 1913 of {\em Lecture Notes in Computer Science}, pages 84--95. Springer, 2000.

\bibitem{chekuri2022}
Chandra Chekuri, Kent Quanrud, and Manuel~R. Torres.
\newblock Densest subgraph: Supermodularity, iterative peeling, and flow.
\newblock In {\em Proceedings of the 2022 Annual ACM-SIAM Symposium on Discrete Algorithms}, SODA, pages 1531--1555. Society for Industrial and Applied Mathematics, 2022.

\bibitem{ChenEtAl15Size}
Wenbin Chen, Lingxi Peng, Jianxiong Wang, Fufang Li, and Maobin Tang.
\newblock Algorithms for the densest subgraph with at least $k$ vertices and with a specified subset.
\newblock In {\em Combinatorial Optimization and Applications}, volume 9486 of {\em Lecture Notes in Computer Science}, pages 566--573. Springer, 2015.

\bibitem{ChenEtAl16Size}
Wenbin Chen, Nagiza~F. Samatova, Matthias~F. Stallmann, William Hendrix, and Weiqin Ying.
\newblock On size-constrained minimum $s$--$t$ cut problems and size-constrained dense subgraph problems.
\newblock {\em Theoretical Computer Science}, 609:434--442, 2016.

\bibitem{danisch2017}
Maximilien Danisch, T.-H.~Hubert Chan, and Mauro Sozio.
\newblock Large scale density-friendly graph decomposition via convex programming.
\newblock In {\em Proceedings of the 26th International Conference on World Wide Web}, WWW '17, pages 233--242, Republic and Canton of Geneva, Switzerland, 2017. International World Wide Web Conferences Steering Committee.

\bibitem{dinic1970algorithm}
E.~A. Dinic.
\newblock Algorithm for solution of a problem of maximum flow in networks with power estimation.
\newblock {\em Soviet Mathematics Doklady}, 11:1277--1280, 1970.

\bibitem{even1975network}
Shimon Even and R.~Endre Tarjan.
\newblock Network flow and testing graph connectivity.
\newblock {\em SIAM Journal on Computing}, 4(4):507--518, 1975.

\bibitem{Farago2019}
András Faragó and Zohre R.~Mojaveri.
\newblock In search of the densest subgraph.
\newblock {\em Algorithms}, 12(8), 2019.

\bibitem{FeigeKortsarzPeleg01}
Uriel Feige, David Peleg, and Guy Kortsarz.
\newblock The dense $k$-subgraph problem.
\newblock {\em Algorithmica}, 29(3):410--421, 2001.

\bibitem{gallo1989}
Giorgio Gallo, Michael~D. Grigoriadis, and Robert~E. Tarjan.
\newblock A fast parametric maximum flow algorithm and applications.
\newblock {\em SIAM Journal on Computing}, 18(1):30--55, 1989.

\bibitem{goldberg1984}
Andrew~V. Goldberg.
\newblock Finding a maximum density subgraph.
\newblock Technical Report UCB/CSD-84-171, University of California, Berkeley, 1984.

\bibitem{Goldberg1990FindingMC}
Andrew~V. Goldberg and Robert~Endre Tarjan.
\newblock Finding minimum-cost circulations by successive approximation.
\newblock {\em Math. Oper. Res.}, 15:430--466, 1990.

\bibitem{GoldsteinLangberg09}
Doron Goldstein and Michael Langberg.
\newblock The dense $k$ subgraph problem, 2009.
\newblock Revised March 2010.

\bibitem{harb2022}
Elfarouk Harb, Kent Quanrud, and Chandra Chekuri.
\newblock Faster and scalable algorithms for densest subgraph and decomposition.
\newblock In S.~Koyejo, S.~Mohamed, A.~Agarwal, D.~Belgrave, K.~Cho, and A.~Oh, editors, {\em Advances in Neural Information Processing Systems}, volume~35, pages 26966--26979. Curran Associates, Inc., 2022.

\bibitem{harbESA2023}
Elfarouk Harb, Kent Quanrud, and Chandra Chekuri.
\newblock Convergence to lexicographically optimal base in a (contra)polymatroid and applications to densest subgraph and tree packing.
\newblock In {\em 31st Annual European Symposium on Algorithms (ESA 2023)}, volume 274 of {\em Leibniz International Proceedings in Informatics (LIPIcs)}, pages 56:1--56:17, 2023.

\bibitem{JonesEtAl23SoS}
Chris Jones, Aaron Potechin, Goutham Rajendran, and Jeff Xu.
\newblock Sum-of-squares lower bounds for densest $k$-subgraph.
\newblock In {\em Proceedings of the 55th Annual ACM Symposium on Theory of Computing}, pages 84--95. Association for Computing Machinery, 2023.

\bibitem{kannan1999analyzing}
Ravindran Kannan and V.~Vinay.
\newblock Analyzing the structure of large graphs.
\newblock Technical report, Forschungsinst. f{\"u}r Diskrete Mathematik, 1999.

\bibitem{Khot06}
Subhash Khot.
\newblock Ruling out {PTAS} for graph min-bisection, dense $k$-subgraph, and bipartite clique.
\newblock {\em SIAM Journal on Computing}, 36(4):1025--1071, 2006.

\bibitem{khuller2009}
Samir Khuller and Barna Saha.
\newblock On finding dense subgraphs.
\newblock In {\em Automata, Languages and Programming}, volume 5555 of {\em Lecture Notes in Computer Science}, pages 597--608. Springer, 2009.

\bibitem{KortsarzPeleg93}
Guy Kortsarz and David Peleg.
\newblock On choosing a dense subgraph.
\newblock In {\em Proceedings of the 34th Annual IEEE Symposium on Foundations of Computer Science}, pages 692--701. IEEE, 1993.

\bibitem{LaekhanukitManurangsiTrabelsi26}
Bundit Laekhanukit, Pasin Manurangsi, and Ohad Trabelsi.
\newblock A note on approximability of densest at-least-$k$-subgraph, 2026.

\bibitem{lanciano2024survey}
Tommaso Lanciano, Atsushi Miyauchi, Adriano Fazzone, and Francesco Bonchi.
\newblock A survey on the densest subgraph problem and its variants.
\newblock {\em ACM Computing Surveys}, 56(8):1--40, 2024.

\bibitem{Lee2010}
Victor~E. Lee, Ning Ruan, Ruoming Jin, and Charu Aggarwal.
\newblock {\em A Survey of Algorithms for Dense Subgraph Discovery}, pages 303--336.
\newblock Springer US, 2010.

\bibitem{Manurangsi17}
Pasin Manurangsi.
\newblock Almost-polynomial ratio {ETH}-hardness of approximating densest $k$-subgraph.
\newblock In {\em Proceedings of the 49th Annual ACM SIGACT Symposium on Theory of Computing}, pages 954--961. Association for Computing Machinery, 2017.

\bibitem{Manurangsi18}
Pasin Manurangsi.
\newblock Inapproximability of maximum biclique problems, minimum $k$-cut and densest at-least-$k$-subgraph from the small set expansion hypothesis.
\newblock {\em Algorithms}, 11(1):10, 2018.

\bibitem{matula1983smallest}
David~W. Matula and Leland~L. Beck.
\newblock Smallest-last ordering and clustering and graph coloring algorithms.
\newblock {\em Journal of the ACM}, 30(3):417--427, 1983.

\bibitem{nguyenEne2024}
Ta~Duy Nguyen and Alina Ene.
\newblock Multiplicative weights update, area convexity and random coordinate descent for densest subgraph problems.
\newblock In {\em Proceedings of the 41st International Conference on Machine Learning}, volume 235 of {\em Proceedings of Machine Learning Research}, pages 37683--37706. PMLR, 2024.

\bibitem{seidman1983network}
Stephen~B. Seidman.
\newblock Network structure and minimum degree.
\newblock {\em Social Networks}, 5(3):269--287, 1983.

\bibitem{SLEATOR1983362}
Daniel~D. Sleator and Robert {Endre Tarjan}.
\newblock A data structure for dynamic trees.
\newblock {\em Journal of Computer and System Sciences}, 26(3):362--391, 1983.

\end{thebibliography}

\appendix

\section{Flows and Blocking Flows} \label{app:flow} This appendix contains the necessary facts and properties of blocking flows used throughout this draft.  

Consider a directed capacitated network
$H=(W,A,c)$ and a feasible $s,t$ flow $f$ where $s, t \in W$ are distinct vertices.  We denote its residual network by $H_f$. The hop-distance of two vertices $u,v \in W$ in $H_f$, $\mathrm{dist}_f(u,v)$, is the minimum number of edges in a $u,v$ path in $H_f$.

The \emph{level graph} $L_f$ is the subgraph of $H_f$ containing precisely the residual arcs $(u,v)$ satisfying
\[
    \mathrm{dist}_f(s,v)=\mathrm{dist}_f(s, u)+1.
\]
In particular, $L_f$ is a shortest path DAG from the source $s$ in $H_f$ in terms of hop-distance.

A feasible flow $g$ in $L_f$, with capacities given by the residual capacities of $H_f$, is a \emph{blocking flow} if every $s,t$ path in $L_f$ contains an arc saturated by $g$.
Equivalently, after augmenting $f$ by $g$, the residual graph $H_{f+g}$ contains no $s$--$t$ path of hop-distance $\mathrm{dist}_f(s,t)$.  This is the standard blocking-flow framework in Dinitz's maximum-flow algorithm
\cite{dinic1970algorithm,even1975network}.

\begin{observation}
\label{obs:blocking-flow-distance}
Let $g$ be a blocking flow in the level graph $L_f$.  Then
\begin{align*}
    \mathrm{dist}_{f+g}(s,t)>\mathrm{dist}_f(s, t).
\end{align*}
\end{observation}

Thus, augmenting a blocking-flow increases the
$s,t$ hop-distance in the residual graph.  In our algorithms, we do not run these phases
until a maximum flow is obtained.  Instead, we stop after a certain
number of phases. We show the absence of short augmenting paths implies useful structural properties about the densest subgraph.  

When $H_f$ has unit capacities, a blocking flow can be constructed in time $O(|W| + |A|)$. If the capacities of all edges in $H_f$ is divisible by a constant $\varepsilon > 0$ and the summation of the capacities is $O(|W| + |A|)$, one can replace each edge $e$ of $A$ with $\frac{c(e)}{\varepsilon}$ copies to obtain an equivalent unit capacity graph. By doing so, we can find a blocking flow in time $O\left( \frac{|W| + |A|}{\varepsilon} \right)$. In general, when the capacities are polynomially bounded in $n$, a blocking flow of $H_f$ can be found in $O \left( |A| \log |W| \right) \cite{Goldberg1990FindingMC, SLEATOR1983362}$.

\end{document}